\documentclass[12pt,a4paper]{article}
\usepackage{amssymb}
\usepackage{amsmath}
\usepackage{setspace}
\usepackage{latexsym}

\newtheorem{theorem}{Theorem}[section]

\newtheorem{corollary}[theorem]{Corollary}

\newtheorem{definition}[theorem]{Definition}
\newtheorem{example}[theorem]{Example}

\newtheorem{lemma}[theorem]{Lemma}

\newtheorem{proposition}[theorem]{Proposition}
\newtheorem{remark}[theorem]{Remark}

\newenvironment{proof}[1][Proof]{\noindent\textbf{#1.} }{\ \rule{0.5em}{0.5em}}
\input{tcilatex}
\begin{document}

\title{ Robust utility maximization for L\'evy processes: Penalization and
solvability}
\author{ {Daniel Hern\'{a}ndez--Hern\'{a}ndez\thanks{%
Centro de Investigaci\'{o}n en Matem\'{a}ticas, Apartado postal 402,
Guanajuato, Gto. 36000, M\'{e}xico. E-mail: dher@cimat.mx } \ \ Leonel P\'{e}%
rez-Hern\'{a}ndez\thanks{%
Departamento de Econom\'{\i}a y Finanzas, Universidad de Guanajuato, DCEA
Campus Guanajuato, C.P. 36250, Guanajuato, Gto. E-mail:
lperezhernandez@yahoo.com}}}
\maketitle

\begin{abstract}
In this paper the robust utility maximization problem for a market model
based on L\'{e}vy processes is analyzed. The interplay between the form of
the utility function and the penalization function required to have a well
posed problem is studied, and for a large class of utility functions it is
proved that the dual problem is solvable as well as the existence of optimal
solutions. The class of equivalent local martingale measures is
characterized in terms of the parameters of the price process, and the
connection with convex risk measures is also presented.
\end{abstract}



\section{Introduction}

\addcontentsline{toc}{section}{Introduction}

The progress in portfolio optimization is operationally related to the
possibility of solving a complex optimization problem in a typically
infinite dimensional space, and the ability to translate the emerging
problem to the available optimization methods. Convex duality, stochastic
control are among the most used approaches.

The development in optimal portfolio management is conceptually determined
by the form of the problems emerging from the prevailing theory of choice
under uncertainty. Optimal portfolio selection corresponds in a very
abstract form to choose a maximal element $X$ with respect to a \textit{%
preference order} from a class of admissible elements $\mathcal{X}.$ In a
very short form the story might be traced as follows: The axiom system
proposed by von Neumann and Morgenstern, and Savage lead to a preference
representation as the expectation of a utility function $U$ under a fixed
probability measure $\mathbb{Q}$. The paradigm of expected utility became
one of the pillars in economics during the last century. Starting from an
expected utility problem of the form 
\begin{equation}
\mathbb{E}_{\mathbb{Q}}\left[ U\left( X\right) \right] \rightarrow \max ,
\label{Introduction_1}
\end{equation}%
Harry Markowitz \cite{Markw 1951} derived in the early 50s, for the first
time, a quantitative solution in form of his celebrated mean-variance
analysis \cite{Markw 1952}, and confronted the academic world with the
ubiquitous trade-off between profit and risk in a financial market. It is
common to refer to (\ref{Introduction_1}) as the Merton-problem, because a
solution to this problem in the context of a continuous time Markovian
market model was established in \cite{Mert 1969} and \cite{Mert 1971} using
stochastic control methods. Harrison and Pliska accomplished in \cite%
{Hrs&Plk 1981} and \cite{Hrs&Plk 1983} the connection to stochastic calculus
(initiated by Bachelier at the beginning of the last century), what led to
the continuous time investment-consumption problems, widely studied in the
second half of the last century.

It is merit of Pliska \cite{Plk 1986} to provide the martingale and duality
approach, which is still one of the most influential ideas to solve the
expected utility maximization problem. For the application of stochastic
control methods in the solution of the dual problem for utility maximization
with consumption see \cite{Net&Dan 2005}. Kramkov and Schachermayer \cite%
{Krk&Scha 1999} and \cite{Krk&Scha 2003} studied the problem (\ref%
{Introduction_1}) in a very general semimartingale setting, for utility
functions defined in the positive halfline. A utility function $U:\left(
0,\infty \right) \longrightarrow \mathbb{R}$ will be hereafter a strictly
increasing, strictly concave, continuously differentiable real function,
which satisfies the Inada conditions (i.e. $U^{\prime }\left( 0+\right)
=+\infty $ and $U^{\prime }\left( \infty -\right) =0$). The log-utility $%
U\left( x\right) =\log \left( x\right) $ and the power utility $U\left(
x\right) =\frac{1}{q}x^{q}$, with $q\in \left( -\infty ,1\right) \backslash
\left\{ 0\right\} $, satisfy those properties, and are in the group of
utility functions that more attention have received in the literature. In 
\cite{Krk&Scha 1999}-\cite{Krk&Scha 2003} the authors fixed a prior
probability measure $\mathbb{Q}$ representing the market measure, and tackle
the primal problem in a dynamic setting for a fixed finite time horizon $T$ 
\begin{equation}
u_{\mathbb{Q}}\left( x\right) :=\sup\limits_{X\in \mathcal{X}\left( x\right)
}\left\{ \mathbb{E}_{\mathbb{Q}}\left[ U\left( X_{T}\right) \right] \right\}
,  \label{Classic Setting: Primal Probl.}
\end{equation}%
over a set of admissible wealth processes $\mathcal{X}\left( x\right) $,
which will be explained later. The market model should be arbitrage free in
the sense that the class of \textit{\ equivalent local martingale measures} $%
\mathcal{Q}_{elmm}\left( \mathbb{Q}\right) :=\{\mathbb{Q}^{\prime }\approx 
\mathbb{Q}:\mathcal{X}\left( 1\right) \subset \mathcal{M}_{loc}\left( 
\mathbb{Q}^{\prime }\right) \}$ is not empty, where $\mathcal{M}_{loc}\left( 
\mathbb{Q}^{\prime }\right) $ denotes the set of local martingales with
respect to $\mathbb{Q}^{\prime }.$ In these papers the analysis was based on
the dual formulation, which basic idea is to pass to the convex conjugate $V$
(also known as the Fenchel-Legendre transformation) of the function $%
-U\left( -x\right) $, defined by 
\begin{equation}
V\left( y\right) =\sup_{x>0}\left\{ U\left( x\right) -xy\right\} ,\qquad y>0.
\label{V convex conjugate of U}
\end{equation}%
From the conditions imposed to the utility function $U$, we have that the
conjugate function $V$ is continuously differentiable, decreasing, and
strictly convex, satisfying: $V^{\prime }\left( 0+\right) =-\infty ,$ $%
V^{\prime }\left( \infty \right) =0,$ $V\left( 0+\right) =U\left( \infty
\right) ,$ $V\left( \infty \right) =U\left( 0+\right) $. Further, the
biconjugate of $U$ is again $U$ itself; in other words the bidual
relationship holds 
\begin{equation*}
U\left( x\right) =\inf_{y>0}\left\{ V\left( y\right) +xy\right\} ,\qquad x>0.
\end{equation*}

Kramkov and Schachermayer \cite{Krk&Scha 1999} formulated the dual problem
in the non-robust setting in terms of the value function 
\begin{equation}
v_{\mathbb{Q}}\left( y\right) :=\inf_{Y\in \mathcal{Y}_{\mathbb{Q}}\left(
y\right) }\left\{ \mathbb{E}_{\mathbb{Q}}\left[ V\left( Y_{T}\right) \right]
\right\} ,  \label{Classical Setting: dual value funct.}
\end{equation}%
where 
\begin{equation}
\begin{array}{rl}
\mathcal{Y}_{\mathbb{Q}}\left( y\right) := & \left\{ Y\geq 0:\text{ }%
Y_{0}=y,YX\ \mathbb{Q}\text{-supermartingale\ }\forall X\in \mathcal{X}%
\left( 1\right) \right\} .%
\end{array}
\label{Def._Yq(y)}
\end{equation}%
Observe that any $Y\in \mathcal{Y}_{\mathbb{Q}}\left( y\right) $ is a $%
\mathbb{Q}$-supermartingale,\ since $X\equiv 1\in \mathcal{X}\left( 1\right)
.$ The authors introduced also in \cite{Krk&Scha 1999} the concept of 
\textit{asymptotic elasticity} $AE\left( U\right) :=\underset{x\rightarrow
\infty }{\lim \sup }\frac{xU^{^{\prime }}\left( x\right) }{U\left( x\right) }
$ and proved that when the utility function $U$ has asymptotic elasticity
strictly less than one $AE\left( U\right) <1$, then:

\begin{enumerate}
\item[(i)] There is always a unique solution for $x>0$, i.e. there exists a
unique $\widehat{X}\in \mathcal{X}\left( x\right) $ such that $u_{\mathbb{Q}%
}\left( x\right) =\mathbb{E}_{\mathbb{Q}}\left[ U\left( \widehat{X}%
_{T}\right) \right] $.

\item[(ii)] The value function $u_{\mathbb{Q}}\left( x\right) $ is a utility
function i.e. strictly increasing, strictly concave, continuously
differentiable and satisfies the Inada conditions ($u^{\prime }\left(
0+\right) =+\infty $ and $u^{\prime }\left( \infty -\right) =0$).

\item[(iii)] The dual problem satisfies $v_{\mathbb{Q}}\left( y\right)
<\infty ,\ \forall y>0$, and it can be restricted to the class of equivalent
local martingale measures $\mathcal{Q}_{elmm}\left( \mathbb{Q}\right) $, 
\begin{equation*}
v_{\mathbb{Q}}\left( y\right) =\inf_{\widetilde{\mathbb{Q}}\in \mathcal{Q}%
_{elmm}\left( \mathbb{Q}\right) }\left\{ \mathbb{E}_{\mathbb{Q}}\left[
V\left( yd\widetilde{\mathbb{Q}}/d\mathbb{Q}\right) \right] \right\} .
\end{equation*}
\end{enumerate}

The previous assertions (i) - (iii) hold when the classical problem $\left( %
\ref{Classic Setting: Primal Probl.}\right) $ is finite for at least some $%
x>0,$ and the non-arbitrage condition $\mathcal{Q}_{elmm}\left( \mathbb{Q}%
\right) \neq \varnothing $ together with the Inada conditions for $U$ are
satisfied. Clearly, the asymptotic elasticity hypothesis involves only the
utility function $U$ and hence such condition is independent of the
financial market.

In a more recent contribution, Kramkov and Schachermayer \cite{Krk&Scha 2003}
proved that a necessary and sufficient condition for (i) - (iii) to hold is
that the \textit{dual function is finite}. Moreover, the authors showed that
the following assertions are equivalent: 
\begin{eqnarray}
v_{\mathbb{Q}}\left( y\right) &<&\infty ,\ \text{for all}\;y>0,
\label{Classic Setting:_Necessary_and_sufficient_condition} \\
\lim_{x\rightarrow \infty }\frac{u_{\mathbb{Q}}\left( x\right) }{x} &=&0, 
\notag \\
\inf_{\widetilde{\mathbb{Q}}\in \mathcal{Q}_{elmm}\left( \mathbb{Q}\right) }%
\mathbb{E}_{\mathbb{Q}}\left[ V\left( yd\widetilde{\mathbb{Q}}/d\mathbb{Q}%
\right) \right] &<&\infty ,\ \text{for all}\;y>0.  \notag
\end{eqnarray}%
When any of these conditions is satisfied, it can be concluded that:

\begin{enumerate}
\item[(iv)] $u_{\mathbb{Q}}\left( x\right) <\infty,\; \text{for all}\; x>0.$

\item[(v)] The primal and dual problems have optimal solutions, $\widehat{X}%
\in \mathcal{X}\left( x\right) $ and $\widehat{Y}\in \mathcal{Y}_{\mathbb{Q}%
}\left( y\right) $ respectively, and are unique. Moreover, for $y=u_{\mathbb{%
Q}}^{\prime }\left( x\right) $ it follows that 
\begin{equation*}
U^{\prime }\left( \widehat{X}_{T}\left( x\right) \right) =\widehat{Y}%
_{T}\left( y\right) .
\end{equation*}

\item[(vi)] The primal and dual value functions, $u_{\mathbb{Q}}\left(
x\right) $ and $v_{\mathbb{Q}}\left( y\right) $ respectively, are conjugate%
\begin{eqnarray*}
u_{\mathbb{Q}}\left( x\right) &=&\inf_{y>0}\left\{ v_{\mathbb{Q}}\left(
y\right) +xy\right\} , \\
v_{\mathbb{Q}}\left( y\right) &=&\sup_{x>0}\left\{ u_{\mathbb{Q}}\left(
x\right) -xy\right\} .
\end{eqnarray*}
\end{enumerate}

Let $\mathcal{Q}$ be the family of probability measures on the measurable
space $\left( \Omega ,\mathcal{F}\right) .$ The choice of the market measure 
$\mathbb{Q\in }\mathcal{Q}$ (model uncertainty or ambiguity) has risen many
empirical studies, and has also motivated (beside some incongruous paradox)
a reexamination of the axiomatic foundations of the theory of choice under
uncertainty. Gilboa and Schmeidler \cite{Gil&Sch 1989} gave a significant
step in this direction, introducing the \textquotedblleft
certainty-independence\textquotedblright\ axiom, what led to robust utility
functionals 
\begin{equation*}
X\longrightarrow \inf_{\mathbb{Q}\in \mathcal{Q}^{\prime }}\left\{ \mathbb{E}%
_{\mathbb{Q}}\left[ U\left( X\right) \right] \right\} ,
\end{equation*}%
where the set of \textquotedblleft prior\textquotedblright\ models $\mathcal{%
Q}^{\prime }\subset \mathcal{Q}$ is assumed to be a convex set of
probability measures on the measurable space $\left( \Omega ,\mathcal{F}%
\right) $. For an overview and details about preference orders and its
(robust) representation see F\"{o}llmer and Schied \cite{FoellSch 2002 a}.
The corresponding robust utility maximization problem 
\begin{equation}
\inf_{\mathbb{Q}\in \mathcal{Q}^{\prime }}\left\{ \mathbb{E}_{\mathbb{Q}}%
\left[ U\left( X\right) \right] \right\} \rightarrow \max ,
\label{Gilboa&Schmeidler_Utility_functional}
\end{equation}%
has being studied by several authors. See \cite{Quenez 2004}, \cite{Gundel
2005}, \cite{Schied & Wu 2005}, \cite{Foell&Gund 2006} and \cite{Gundel 2006}
and references therein.

A natural observation is that the worst case approach in (\ref%
{Gilboa&Schmeidler_Utility_functional}) does not discriminate among all
possible models in $\mathcal{Q}^{\prime }$, what again is reflected in
inconsistencies in the axiomatic system proposed in \cite{Gil&Sch 1989}.
Maccheroni, Marinacci and Rustichini \cite{MMR 2006} proposed a relaxed
axiomatic system, which led to utility functionals 
\begin{equation}
X\longrightarrow \inf_{\mathbb{Q}\in \mathcal{Q}^{\prime }}\left\{ \mathbb{E}%
_{\mathbb{Q}}\left[ U\left( X\right) \right] +\vartheta \left( \mathbb{Q}%
\right) \right\} ,  \label{MMR_functionals}
\end{equation}%
where the penalty function $\vartheta $ assigns a weight $\vartheta \left( 
\mathbb{Q}\right) $ to each model $\mathbb{Q}\in \mathcal{Q}^{\prime }$.
Such preference representations take into account both the risk preferences
and model uncertainty. Schied \cite{Schd 2007} developed the corresponding
dual theory for utility functions defined in the positive halfline and
utility functionals of the form $\left( \ref{MMR_functionals}\right) $. The
goal of the economic agent, with an initial capital $x>0,$ will be now to
maximize the penalized expected utility from a terminal wealth in the worst
case model. This means that the agent seeks to solve the associated robust
expected utility problem with value function 
\begin{equation}
u\left( x\right) :=\sup\limits_{X\in \mathcal{X}\left( x\right)
}\inf\limits_{\mathbb{Q}\in \mathcal{Q}_{\ll }^{\vartheta }\left( \mathbb{P}%
\right) }\left\{ \mathbb{E}_{\mathbb{Q}}\left[ U\left( X_{T}\right) \right]
+\vartheta \left( \mathbb{Q}\right) \right\} ,
\label{robust_probl._primal_value_funct}
\end{equation}%
where $\mathcal{Q}_{\ll }^{\vartheta }:=\left\{ \mathbb{Q}\ll \mathbb{P}%
:\vartheta \left( \mathbb{Q}\right) <\infty \right\} $ for a fixed reference
measure $\mathbb{P}.$ To guarantee that the $\mathbb{Q}$-expectation is well
defined, we extend the operator $\mathbb{E}_{\mathbb{Q}}\left[ U\left( \cdot
\right) \right] $ to $\mathcal{L}^{0}$, as in Schied \cite[p. 111]{Schd 2007}%
, by 
\begin{equation}
\mathbb{E}_{\mathbb{Q}}\left[ X\right] :=\sup_{n\in \mathbb{N}}\mathbb{E}_{%
\mathbb{Q}}\left[ X\wedge n\right] =\lim_{n\rightarrow \infty }\mathbb{E}_{%
\mathbb{Q}}\left[ X\wedge n\right] \quad X\in \mathcal{L}^{0}\left( \Omega ,%
\mathcal{F}\right) .  \label{Eq[X]_extention}
\end{equation}%
The corresponding dual value function is defined by 
\begin{equation}
v\left( y\right) :=\inf_{\mathbb{Q}\in \mathcal{Q}_{\ll }^{\vartheta
}}\left\{ v_{\mathbb{Q}}\left( y\right) +\vartheta \left( \mathbb{Q}\right)
\right\} .  \label{robust_probl_dual_value_funct}
\end{equation}%
In this robust setting the necessary and sufficient condition $\left( \ref%
{Classic Setting:_Necessary_and_sufficient_condition}\right) $ is
transformed into 
\begin{equation}
v_{\mathbb{Q}}\left( y\right) <\infty \quad \text{for all}\;\mathbb{Q\in }%
\mathcal{Q}_{\approx }^{\vartheta }\quad \text{and }\;y>0,
\label{vQ(y)<oo_for_y>0_Q_nice}
\end{equation}%
where $\mathcal{Q}_{\approx }^{\vartheta }:=\left\{ \mathbb{Q}\approx 
\mathbb{P}:\vartheta \left( \mathbb{Q}\right) <\infty \right\} .$

\begin{remark}
\label{V<0=>Solution exist} When the conjugate convex function $V$ is
bounded from above it follows immediately that the penalized robust utility
maximization problem (\ref{robust_probl._primal_value_funct}) has a solution
for any proper penalty function $\vartheta $. This is the case, for
instance, of the power utility function $U\left( x\right) :=\frac{1}{q}x^{q}$%
, for $q\in \left( -\infty ,0\right) $, where the convex conjugate function $%
V\left( x\right) =\frac{1}{p}x^{-p}\leq 0$, with $p:=\frac{q}{1-q}$.
Moreover, condition (\ref{vQ(y)<oo_for_y>0_Q_nice}) points out also that the
existence of solution to the problem (\ref{robust_probl._primal_value_funct}%
) relies on the positive part $V^{+}$ of the convex conjugate function.
\end{remark}

Let $\vartheta $ be a penalty function bounded from below, which corresponds
to the minimal penalty function of a normalized and sensitive convex risk
measure, see Section \ref{Sub_Sect:_Static_Risk_Measures} for details and
further references. Assuming condition $\left( \ref{vQ(y)<oo_for_y>0_Q_nice}%
\right) $, the following assertions hold for the robust problem $\left( \ref%
{robust_probl._primal_value_funct}\right) $.

\begin{enumerate}
\item[(vii)] The robust value function $u\left( x\right) $ is strictly
concave and takes only finite values.

\item[(viii)] The \textquotedblleft minimax property\textquotedblright\ is
satisfied 
\begin{equation*}
\sup\limits_{X\in \mathcal{X}\left( x\right) }\inf\limits_{\mathbb{Q}\in 
\mathcal{Q}_{\ll }^{\vartheta }}\left\{ \mathbb{E}_{\mathbb{Q}}\left[
U\left( X_{T}\right) \right] +\vartheta \left( \mathbb{Q}\right) \right\}
=\inf\limits_{\mathbb{Q}\in \mathcal{Q}_{\ll }^{\vartheta
}}\sup\limits_{X\in \mathcal{X}\left( x\right) }\left\{ \mathbb{E}_{\mathbb{Q%
}}\left[ U\left( X_{T}\right) \right] +\vartheta \left( \mathbb{Q}\right)
\right\} ;
\end{equation*}%
in other words, 
\begin{equation*}
u\left( x\right) =\inf\limits_{\mathbb{Q}\in \mathcal{Q}_{\ll }^{\vartheta
}}\left\{ u_{\mathbb{Q}}\left( x\right) +\vartheta \left( \mathbb{Q}\right)
\right\} .
\end{equation*}

\item[(ix)] $u$ and $v$ are conjugate 
\begin{equation*}
u\left( x\right) =\inf_{y>0}\left( v\left( y\right) +xy\right) \text{\qquad
and\qquad }v\left( y\right) =\sup_{x>0}\left( u\left( x\right) -xy\right) .
\end{equation*}

\item[(x)] $v$ is convex, continuously differentiable, and take only finite
values.

\item[(xi)] The dual problem $\left( \ref{robust_probl_dual_value_funct}%
\right) $ has an optimal solution. That is, there exist $\mathbb{Q}^{\ast
}\in \mathcal{Q}_{\ll }^{\vartheta }$ and $Y^{\ast }\in \mathcal{Y}_{\mathbb{%
Q}^{\ast }}\left( y\right) $ such that 
\begin{equation*}
\mathbb{E}_{\mathbb{Q}^{\ast }}\left[ V\left( Y_{T}^{\ast }\right) \right]
+\vartheta \left( \mathbb{Q}^{\ast }\right) =\inf\limits_{\mathbb{Q}\in 
\mathcal{Q}_{\ll }^{\vartheta }}\left\{ \inf\limits_{Y\in \mathcal{Y}_{%
\mathbb{Q}}\left( y\right) }\left\{ \mathbb{E}_{\mathbb{Q}}\left[ V\left(
Y_{T}\right) \right] \right\} +\vartheta \left( \mathbb{Q}\right) \right\} ,
\end{equation*}%
which is maximal in the sense that any other solution $\left( \mathbb{Q}%
,Y\right) $ satisfies $\mathbb{Q}\ll \mathbb{Q}^{\ast }$ and $%
Y_{T}=Y_{T}^{\ast }\ \mathbb{Q}$-a.s. .

\item[(xii)] For each $x>0$ there exists an optimal solution $X^{\ast }\in 
\mathcal{X}\left( x\right) $ to the robust problem $\left( \ref%
{robust_probl._primal_value_funct}\right) .$ Furthermore, let $y>0$, such
that $v^{\prime }\left( y\right) =-x$, and $\left( \mathbb{Q}^{\ast
},Y^{\ast }\right) $ be a solution to the dual problem $\left( \ref%
{robust_probl_dual_value_funct}\right) $. Then $\left( \mathbb{Q}^{\ast
},X^{\ast }\right) $, with%
\begin{equation*}
X_{T}^{\ast }:=-V^{^{\prime }}\left( Y_{T}^{\ast }\right) ,
\end{equation*}%
is a saddlepoint for the robust problem 
\begin{equation*}
u\left( x\right) =\mathbb{E}_{\mathbb{Q}^{\ast }}\left[ U\left( X_{T}^{\ast
}\right) \right] +\vartheta \left( \mathbb{Q}^{\ast }\right) =\inf\limits_{%
\mathbb{Q}\in \mathcal{Q}_{\ll }^{\vartheta }}\sup\limits_{X\in \mathcal{X}%
\left( x\right) }\left\{ \mathbb{E}_{\mathbb{Q}}\left[ U\left( X_{T}\right) %
\right] +\vartheta \left( \mathbb{Q}\right) \right\} .
\end{equation*}
\end{enumerate}

The outline and description of the main contributions of the paper are as
follows: In Section \ref{Sect. Preliminaries} we propose the probability
space on which we shall develop our work, and describe the class of
absolutely continuous probabilities with respect to a reference probability
measure $\mathbb{P}.$ We also recall some fundamental facts about static
convex measures of risk needed to establish the main results.

Samuelson \cite{Saml 1965} seems to be the first to propose a geometric
Brownian motion as a model for the prices of the underlying assets in a
market; it is often referred (wrongly) as the Black \& Scholes model. This
idea led to the, almost ubiquitous, exponential semimartingales models. We
use one of them to introduce the market model in Section \ref{Sect.
Market_Model}, which need not to have independent increments but include
certain L\'{e}vy exponential models, and has been used to study some
problems close to ours; see for instance \cite{Mata & Oks 2008} and \cite%
{Oks&Sul}. We also give in this section a characterization of the equivalent
local martingale measures for the proposed model. This contribution extends
to our setting a result of Kunita \cite{Kunita 2004} for L\'{e}vy
exponential models. We finish this section introducing a family of
penalties, which are minimal for the convex measures of risk generated by
duality.

Once we have introduced necessary conditions for the penalization and the
corresponding convex measure of risk $\rho ,$ which are relevant to develop
the duality theory for the maximization of a penalized robust expected
utility problem as in Schied \cite{Schd 2007}, we address in Section \ref%
{Sect Robust Utility Maximization} the relationship between the choice of a
penalty function and the existence of a solution to the dual problem. For
the power and the logarithmic utility functions we provide, in each case,
thresholds for the family of penalty functions, which guarantee the
existence of solutions to the optimal allocation problem. These results are
the main contributions of this work and their proof are based on Theorem \ref%
{theta=minimal penalty function} in Section \ref{MinimalPenalties}. For
stochastic volatility models, the robust utility maximization problem was
addressed in \cite{Hdz&Sch_2006} and \cite{Hdz&Sch} using stochastic control
technics. We finish this section with a representation of the dual problem,
given in Theorem \ref{Model_1_robust dual value function_(Thm)}, in terms of
certain coefficients for an arbitrary utility function.

\section{Preliminaries\label{Sect. Preliminaries}}

\setcounter{equation}{0} Within a probability space which supports a
semimartingale with the weak predictable representation property, there is a
representation of the density processes of the absolutely continuous
probability measures by means of two coefficients. Roughly speaking, the
weak predictable representation property means that the \textquotedblleft
dimension\textquotedblright\ of the linear space of local martingales is
two. Throughout these coefficients we can represent every local martingale
as a combination of two components, namely an stochastic integral with
respect to the continuous part of the semimartingale and an integral with
respect to its compensated jump measure. This is of course the case for
local martingales, and with more reason this observation about the
dimensionality holds for the martingales associated with the corresponding
densities processes. In this section we also review some concepts of
stochastic calculus needed to understand these representation properties.

\subsection{Fundamentals of L\'{e}vy and semimartingales processes \label%
{Sub_Sect:_Fundamentals_Levy_and_Semimartingales}}

Let $\left( \Omega ,\mathcal{F},\mathbb{P}\right) $ be a probability space.
We say that $L:=\left\{ L_{t}\right\} _{t\in \mathbb{R}_{+}}$ is a L\'{e}vy
process for this probability space if it is an adapted c\'{a}dl\'{a}g
process with independent stationary increments starting at zero. The
filtration considered is $\mathbb{F}:=\left\{ \mathcal{F}_{t}^{\mathbb{P}%
}\left( L\right) \right\} _{t\in \mathbb{R}_{+}}$, the completion of its
natural filtration, i.e. $\mathcal{F}_{t}^{\mathbb{P}}\left( L\right)
:=\sigma \left\{ L_{s}:s\leq t\right\} \vee \mathcal{N}$ where $\mathcal{N}$
is the $\sigma $-algebra generated by all $\mathbb{P}$-null sets. The jump
measure of $L$ is denoted by $\mu :\Omega \times \left( \mathcal{B}\left( 
\mathbb{R}_{+}\right) \otimes \mathcal{B}\left( \mathbb{R}_{0}\right)
\right) \rightarrow \mathbb{N}$ where $\mathbb{R}_{0}:=\mathbb{R}\setminus
\left\{ 0\right\} $. The dual predictable projection of this measure, also
known as its L\'{e}vy system, satisfies the relation $\mu ^{\mathcal{P}%
}\left( dt,dx\right) =dt\times \nu \left( dx\right) $, where $\nu \left(
\cdot \right) :=\mathbb{E}\left[ \mu \left( \left[ 0,1\right] \times \cdot
\right) \right] $ is the, so called, intensity or L\'{e}vy measure of $L.$

The L\'{e}vy-It\^{o} decomposition of $L$ is given by 
\begin{equation}
L_{t}=bt+W_{t}+\int\limits_{\left[ 0,t\right] \times \left\{ 0<\left\vert
x\right\vert \leq 1\right\} }x\left\{ \mu \left( ds,dx\right) -\nu \left(
dx\right) ds\right\} +\int\limits_{\left[ 0,t\right] \times \left\{
\left\vert x\right\vert >1\right\} }x\mu \left( ds,dx\right) .
\label{Levy-Ito_decomposition}
\end{equation}%
It implies that $L^{c}=W$ is the Wiener process, and hence $\left[ L^{c}%
\right] _{t}=t$, where $\left( \cdot \right) ^{c}$ and $\left[ \,\cdot \,%
\right] $ denote the continuous martingale part and the process of quadratic
variation of any semimartingale, respectively. For the predictable quadratic
variation we use the notation $\left\langle \,\cdot \,\right\rangle $.

Even though most of the paper deals with L\'{e}vy processes, we need to
introduce some notation from the theory of semimartingales, and present some
results needed in the next sections. Denote by $\mathcal{V}$ the set of c%
\'{a}dl\'{a}g, adapted processes with finite variation, and let $\mathcal{V}%
^{+}\subset \mathcal{V}$ be the subset of non-decreasing processes in $%
\mathcal{V}$ starting at zero.

Let $\mathcal{A}\subset \mathcal{V}$ be the class of processes with
integrable variation, i.e. $A\in \mathcal{A}$ if and only if $%
\bigvee\nolimits_{0}^{\infty }A\in L^{1}\left( \mathbb{P}\right) $, where $%
\bigvee\nolimits_{0}^{t}A$ denotes the variation of $A$ over the finite
interval $\left[ 0,t\right] $. The subset $\mathcal{A}^{+}\subset \mathcal{A}
$ represents those processes which are also increasing i.e. with
non-negative right-continuous increasing trajectories. Furthermore, $%
\mathcal{A}_{loc}$ (resp. $\mathcal{A}_{loc}^{+}$) is the collection of
adapted processes with locally integrable variation (resp. adapted locally
integrable increasing processes). For a c\'{a}dl\'{a}g process $X$ we denote
by $X_{-}:=\left( X_{t-}\right) $ the left hand limit process, with $%
X_{0-}:=X_{0}$ by convention, and by $\bigtriangleup X=\left( \bigtriangleup
X_{t}\right) $ the jump process $\bigtriangleup X_{t}:=X_{t}-X_{t-}$.

Given an adapted c\'{a}dl\'{a}g semimartingale $U$, the jump measure and its
dual predictable projection (or compensator) are denoted by $\mu _{U}\left( %
\left[ 0,t\right] \times A\right) :=\sum_{s\leq t}\mathbf{1}_{A}\left(
\triangle U_{s}\right) $ and $\mu _{U}^{\mathcal{P}}$, respectively.
Further, we denote by $\mathcal{P}\subset \mathcal{F}\otimes \mathcal{B}%
\left( \mathbb{R}_{+}\right) $ the predictable $\sigma $-algebra and by $%
\widetilde{\mathcal{P}}:=\mathcal{P}\otimes \mathcal{B}\left( \mathbb{R}%
_{0}\right) .$ With some abuse of notation, we write $\theta _{1}\in 
\widetilde{\mathcal{P}}$ when the function $\theta _{1}:$ $\Omega \times 
\mathbb{R}_{+}\times \mathbb{R}_{0}\rightarrow \mathbb{R}$ is $\widetilde{%
\mathcal{P}}$-measurable and $\theta \in \mathcal{P}$ for predictable
processes.

Let 
\begin{equation}
\begin{array}{clc}
\mathcal{L}\left( U^{c}\right) := & \left\{ \theta \in \mathcal{P}:\exists
\left\{ \tau _{n}\right\} _{n\in \mathbb{N}}\text{ sequence of stopping
times with }\tau _{n}\uparrow \infty \right. &  \\ 
& \left. \text{and }\mathbb{E}\left[ \int\nolimits_{0}^{\tau _{n}}\theta
^{2}d\left[ U^{c}\right] \right] <\infty \ \forall n\in \mathbb{N}\right\} & 
\end{array}
\label{Def._L(U)}
\end{equation}%
be the class of predictable processes $\theta \in \mathcal{P}$ integrable
with respect to $U^{c}$ in the sense of local martingale, and by 
\begin{equation*}
\Lambda \left( U^{c}\right) :=\left\{ \int \theta _{0}dU^{c}:\theta _{0}\in 
\mathcal{L}\left( U^{c}\right) \right\}
\end{equation*}%
the linear space of processes which admit a representation as the stochastic
integral with respect to $U^{c}$. For an integer valued random measure $%
\widetilde{\mu }$ we denote by $\mathcal{G}\left( \widetilde{\mu }\right) $
the class of $\widetilde{\mathcal{P}}$-measurable processes $\theta _{1}:$ $%
\Omega \times \mathbb{R}_{+}\times \mathbb{R}_{0}\rightarrow \mathbb{R}$
satisfying the following conditions: 
\begin{equation*}
\begin{array}{cl}
\left( i\right) & \theta _{1}\in \widetilde{\mathcal{P}}, \\ 
\left( ii\right) & \int\nolimits_{\mathbb{R}_{0}}\left\vert \theta
_{1}\left( t,x\right) \right\vert \widetilde{\mu }^{\mathcal{P}}\left(
\left\{ t\right\} ,dx\right) <\infty \ \forall t>0, \\ 
\left( iii\right) & \text{The process } \\ 
& \left\{ \sqrt{\sum\limits_{s\leq t}\left\{ \int\nolimits_{\mathbb{R}%
_{0}}\theta _{1}\left( s,x\right) \widetilde{\mu }\left( \left\{ s\right\}
,dx\right) -\int\nolimits_{\mathbb{R}_{0}}\theta _{1}\left( s,x\right) 
\widetilde{\mu }^{\mathcal{P}}\left( \left\{ s\right\} ,dx\right) \right\}
^{2}}\right\} _{t\in \mathbb{R}_{+}}\in \mathcal{A}_{loc}^{+}.%
\end{array}%
\end{equation*}%
The set $\mathcal{G}\left( \widetilde{\mu }\right) $ represents the domain
of the functional $\theta _{1}\rightarrow \int \theta _{1}d\left( \widetilde{%
\mu }-\widetilde{\mu }^{\mathcal{P}}\right) .$ We use the notation $\int
\theta _{1}d\left( \widetilde{\mu }-\widetilde{\mu }^{\mathcal{P}}\right) $
to write the value of this functional in $\theta _{1}$. It is important to
point out that this integral functional is not, in general, the integral
with respect to the difference of two measures. But for $\theta _{1}\in 
\widetilde{\mathcal{P}}$ with $\tint\nolimits_{\left[ 0,t\right] \times 
\mathbb{R}_{0}}\theta _{1}d\mu \in \mathcal{A}_{loc}$ we have $\theta
_{1}\in \mathcal{G}\left( \widetilde{\mu }\right) $ and 
\begin{equation*}
\tint \theta _{1}\left( t,x\right) d\left\{ \widetilde{\mu }-\widetilde{\mu }%
^{\mathcal{P}}\right\} =\tint \theta _{1}\left( t,x\right) \widetilde{\mu }%
\left( dt,dx\right) -\tint \theta _{1}\left( t,x\right) \widetilde{\mu }^{%
\mathcal{P}}\left( dt,dx\right) .
\end{equation*}%
For a detailed exposition on these topics see He, Wang and Yan \cite%
{HeWanYan} or Jacod and Shiryaev \cite{Jcd&Shry 2003}, which are our basic
references.

In particular, for the L\'{e}vy process $L$ with jump measure $\mu $, 
\begin{equation}
\mathcal{G}\left( \mu \right) \equiv \left\{ \theta _{1}\in \widetilde{%
\mathcal{P}}:\left\{ \sqrt{\sum\nolimits_{s\leq t}\left\{ \theta _{1}\left(
s,\triangle L_{s}\right) \right\} ^{2}\mathbf{1}_{\mathbb{R}_{0}}\left(
\triangle L_{s}\right) }\right\} _{t\in \mathbb{R}_{+}}\in \mathcal{A}%
_{loc}^{+}\right\} ,  \label{G(miu) Definition}
\end{equation}%
since $\mu ^{\mathcal{P}}\left( \left\{ t\right\} \times A\right) =0$, for
any Borel set $A$ of $\mathbb{R}_{0}$. Recall also that for an adapted
process of finite variation $A\in \mathcal{V}$ we have 
\begin{equation}
A\in \mathcal{A}_{loc}\Longleftrightarrow \sqrt{\tsum\nolimits_{s\leq \cdot
}\left( \bigtriangleup A_{s}\right) ^{2}}\in \mathcal{A}_{loc}^{+}.
\label{Aloc_iff_(3)}
\end{equation}%
Therefore for $\theta _{1}\in \mathcal{G}\left( \mu \right) $ with $%
\tint\nolimits_{\left[ 0,t\right] \times \mathbb{R}_{0}}\left\vert \theta
_{1}\right\vert d\mu <\infty $ $\forall t\ \mathbb{P}$-a.s. it follows that $%
\tint\nolimits_{\left[ 0,t\right] \times \mathbb{R}_{0}}\theta _{1}d\mu \in 
\mathcal{A}_{loc}.$ Furthemore using a localizing argument we have for $%
\theta _{1},\theta _{1}^{\prime }\in \mathcal{G}\left( \mu \right) $ with $%
\left\{ \theta _{1}^{\prime }\left( t,\triangle L_{t}\right) \right\} _{t}$
a locally bounded process that $\tint\nolimits_{\left[ 0,t\right] \times 
\mathbb{R}_{0}}\left\vert \theta _{1}\theta _{1}^{\prime }\right\vert d\mu
\in \mathcal{A}_{loc}^{+}.$

We say that the semimartingale $U$ has the \textit{weak property of
predictable representation} when 
\begin{equation}
\mathcal{M}_{loc,0}=\Lambda \left( U^{c}\right) +\left\{ \int \theta
_{1}d\left( \mu _{U}-\mu _{U}^{\mathcal{P}}\right) :\theta _{1}\in \mathcal{G%
}\left( \mu _{U}\right) \right\} ,\   \label{Def_weak_predictable_repres.}
\end{equation}%
where the previous sum is the linear sum of the vector spaces, and $\mathcal{%
M}_{loc,0}$ is the linear space of local martingales starting at zero.

The integral representation of a semimartingale $U$ asserts that 
\begin{equation}
U_{t}=U_{0}+\alpha _{t}^{U}+U_{t}^{c}+\int\limits_{\left[ 0,t\right] \times
\left\{ 0<\left\vert x\right\vert \leq 1\right\} }x\left\{ \mu _{U}\left(
ds,dx\right) -\mu _{U}^{\mathcal{P}}\left( dx,ds\right) \right\}
+\int\limits_{\left[ 0,t\right] \times \left\{ \left\vert x\right\vert
>1\right\} }x\mu _{U}\left( ds,dx\right) ,
\label{Integral_represnt._of_Semimartingale}
\end{equation}%
where $\alpha _{t}^{U}$ is a predictable process with finite variation and $%
\alpha _{0}^{U}=0.$ Taking $\beta _{t}^{U}:=\left[ U^{c}\right] _{t}$ we
define $\left( \alpha ^{U},\beta ^{U},\mu _{U}^{\mathcal{P}}\right) $ as the 
\textit{predictable characteristics} (\textit{predictable triplet}, \textit{%
local characteristics}) of the semimartingale $U.$

\subsection{Density processes \label{Sub_Sect:_Density_processes}}

Given an absolutely continuous probability measure $\mathbb{Q}\ll \mathbb{P}$
in a filtered probability space, where a semimartingale with the weak
predictable representation property is defined, the structure of the density
process has been studied extensively by several authors; see Theorem 14.41
in He, Wang and Yan \cite{HeWanYan} or Theorem III.5.19 in Jacod and
Shiryaev \cite{Jcd&Shry 2003}.

It is well known that the L\'{e}vy-processes satisfy the weak property of
predictable representation when the completed natural filtration is
considered. In the following lemma we present the characterization of the
density processes for the case of these processes. For L\'{e}vy processes
the proof can be found in \cite{HH-PH}.

\begin{lemma}
\label{Q<<P =>} Given an absolutely continuous probability measure $\mathbb{Q%
}\ll \mathbb{P}$, there exist coefficients $\theta _{0}\in \mathcal{L}\left(
W\right) \ $and $\theta _{1}\in \mathcal{G}\left( \mu \right) $ such that 
\begin{equation*}
\frac{d\mathbb{Q}_{t}}{d\mathbb{P}_{t}}=\mathcal{E}\left( Z^{\theta }\right)
\left( t\right) ,
\end{equation*}%
where 
\begin{equation}
Z_{t}^{\theta }:=\int\nolimits_{]0,t]}\theta
_{0}dW+\int\nolimits_{]0,t]\times \mathbb{R}_{0}}\theta _{1}\left(
s,x\right) \left( \mu \left( ds,dx\right) -ds\ \nu \left( dx\right) \right) ,
\label{Def._Ztheta(t)}
\end{equation}%
and $\mathcal{E}$ represents the Doleans-Dade exponential of a
semimartingale. The coefficients $\theta _{0}$ and $\theta _{1}$ are unique, 
$\mathbb{P}$-a.s. and $\mu _{\mathbb{P}}^{\mathcal{P}}\left( ds,dx\right) $%
-a.s., respectively.
\end{lemma}

For $\mathbb{Q}\ll \mathbb{P}$ the function $\theta _{1}\left( \omega
,t,x\right) $ described in Lemma \ref{Q<<P =>} determines the density of the
predictable projection $\mu _{\mathbb{Q}}^{\mathcal{P}}\left( dt,dx\right) $
with respect to $\mu _{\mathbb{P}}^{\mathcal{P}}\left( dt,dx\right) $ (see
He,Wang and Yan \cite{HeWanYan} or Jacod and Shiryaev \cite{Jcd&Shry 2003}).
More precisely, for\ $B\in \left( \mathcal{B}\left( \mathbb{R}_{+}\right)
\otimes \mathcal{B}\left( \mathbb{R}_{0}\right) \right) $ we have 
\begin{equation}
\mu _{\mathbb{Q}}^{\mathcal{P}}\left( \omega ,B\right) =\int_{B}\left(
1+\theta _{1}\left( \omega ,t,x\right) \right) \mu _{\mathbb{P}}^{\mathcal{P}%
}\left( dt,dx\right) .  \label{Q<<P=>_miu_wrt_Q}
\end{equation}

In what follows we restrict ourself to the time interval $\left[ 0,T\right]
, $ for some $T>0$ fixed, and take $\mathcal{F}=\mathcal{F}_{T}.$ We denote
by $\mathcal{Q}_{\ll }(\mathbb{P})$ the subclass of absolutely continuous
probability measure with respect to $\mathbb{P}$ and by $\mathcal{Q}%
_{\approx }\left( \mathbb{P}\right) $ the subclass of equivalent probability
measures. The corresponding classes of density processes associated to $%
\mathcal{Q}_{\ll }(\mathbb{P})$ and $\mathcal{Q}_{\approx }\left( \mathbb{P}%
\right) $ are denoted by $\mathcal{D}_{\ll }\left( \mathbb{P}\right) $ and $%
\mathcal{D}_{\approx }\left( \mathbb{P}\right) $, respectively. For
instance, in the former case 
\begin{equation}
\mathcal{D}_{\ll }\left( \mathbb{P}\right) :=\left\{ D=\left\{ D_{t}\right\}
_{t\in \left[ 0,T\right] }:\exists \mathbb{Q}\in \mathcal{Q}_{\ll }\left( 
\mathbb{P}\right) \text{ with }D_{t}=\left. \frac{d\mathbb{Q}}{d\mathbb{P}}%
\right\vert _{\mathcal{F}_{t}}\right\} ,  \label{Def._D<<}
\end{equation}%
and the processes in this set are of the form 
\begin{equation}
\begin{array}{rl}
D_{t}= & \exp \left\{ \int\limits_{]0,t]}\theta
_{0}dW+\int\limits_{]0,t]\times \mathbb{R}_{0}}\theta _{1}\left( s,x\right)
\left( \mu \left( ds,dx\right) -\nu \left( dx\right) ds\right) -\frac{1}{2}%
\int\nolimits_{]0,t]}\left( \theta _{0}\right) ^{2}ds\right\} \times \\ 
& \times \exp \left\{ \int\limits_{]0,t]\times \mathbb{R}_{0}}\left\{ \ln
\left( 1+\theta _{1}\left( s,x\right) \right) -\theta _{1}\left( s,x\right)
\right\} \mu \left( ds,dx\right) \right\} ,%
\end{array}
\label{D(t) explicita}
\end{equation}%
for $\theta _{0}\in \mathcal{L}\left( W\right) $ and $\theta _{1}\in 
\mathcal{G}\left( \mu \right) $.

If $\int \theta _{1}\left( s,x\right) \mu \left( ds,dx\right) \in \mathcal{A}%
_{loc}\left( \mathbb{P}\right) $ the previous formula can be written as 
\begin{eqnarray}
D_{t} &=&\exp \left\{ \int\limits_{]0,t]}\theta _{0}dW-\frac{1}{2}%
\int\limits_{]0,t]}\left( \theta _{0}\left( s\right) \right)
^{2}ds+\int\limits_{]0,t]\times \mathbb{R}_{0}}\ln \left( 1+\theta
_{1}\left( s,x\right) \right) \mu \left( ds,dx\right) \right.
\label{Density_for_Aloc} \\
&&\left. -\int\limits_{]0,t]\times \mathbb{R}_{0}}\theta _{1}\left(
s,x\right) \nu \left( dx\right) ds\right\} .  \notag
\end{eqnarray}

\subsection{Static measures of risk\label{Sub_Sect:_Static_Risk_Measures}}

Let $X:\Omega \rightarrow \mathbb{R}$ be a mapping from a set $\Omega $ of
possible market scenarios, representing the discounted net worth of the
position. Uncertainty is represented by the measurable space $(\Omega ,%
\mathcal{F})$, and we denote by $\mathcal{X}$ the linear space of bounded
financial positions, including constant functions.

\begin{definition}
The function $\rho :\mathcal{X}\rightarrow \mathbb{R}$, quantifying the risk
of $X$, is a \textit{monetary risk measure} if it satisfies the following
properties: 
\begin{equation}
\begin{array}{rl}
\text{Monotonicity:} & \text{If }X\leq Y\text{ then }\rho \left( X\right)
\geq \rho \left( Y\right) \ \forall X,Y\in \mathcal{X}.%
\end{array}
\label{Monotonicity}
\end{equation}%
$\smallskip \ $%
\begin{equation}
\begin{array}{rl}
\text{Translation Invariance:} & \rho \left( X+a\right) =\rho \left(
X\right) -a\ \forall a\in \mathbb{R}\ \forall X\in \mathcal{X}.%
\end{array}
\label{Translation Invariance}
\end{equation}%
When this function satisfies also the convexity property 
\begin{equation}
\begin{array}{rl}
& \rho \left( \lambda X+\left( 1-\lambda \right) Y\right) \leq \lambda \rho
\left( X\right) +\left( 1-\lambda \right) \rho \left( Y\right) \ \forall
\lambda \in \left[ 0,1\right] \ \forall X,Y\in \mathcal{X},%
\end{array}
\label{Convexity}
\end{equation}%
it is said that $\rho $ is a convex risk measure.
\end{definition}

We say that a set function $\mathbb{Q}:\mathcal{F}\rightarrow \left[ 0,1%
\right] $ is a \textit{probability content} if it is finite additive and $%
\mathbb{Q}\left( \Omega \right) =1$. The set of \textit{probability contents}
on this measurable space is denoted by $\mathcal{Q}_{cont}$. From the
general theory of static convex risk measures, we know that any map $\psi :%
\mathcal{Q}_{cont}\rightarrow \mathbb{R}\cup \{+\infty \},$ with $%
\inf\nolimits_{\mathbb{Q}\in \mathcal{Q}_{cont}}\psi (\mathbb{Q})\in \mathbb{%
R}$, induces a static convex measure of risk as a mapping $\rho :\mathfrak{M}%
_{b}\rightarrow \mathbb{R}$ given by 
\begin{equation}
\rho (X):=\sup\nolimits_{\mathbb{Q}\in \mathcal{Q}_{cont}}\left\{ \mathbb{E}%
_{\mathbb{Q}}\left[ -X\right] -\psi (\mathbb{Q})\right\} .
\label{Static_CMR_induced_by_phi}
\end{equation}%
Here $\mathfrak{M}$ denotes the class of measurable functions and $\mathfrak{%
M}_{b}$ the subclass of bounded measurable functions. F\"{o}llmer and Schied 
\cite[Theorem 3.2]{FoellSch 2002 b} and Frittelli and Rosazza Gianin \cite[%
Corollary 7]{FritRsza 2002} proved that any convex risk measure is
essentially of this form.

More precisely, a convex measure of risk $\rho $ on the space of bounded
functions $\mathfrak{M}_{b}\left( \Omega ,\mathcal{F}\right) $ has the
representation 
\begin{equation}
\rho (X)=\sup\limits_{\mathbb{Q}\in \mathcal{Q}_{cont}}\left\{ \mathbb{E}_{%
\mathbb{Q}}\left[ -X\right] -\psi _{\rho }^{\ast }\left( \mathbb{Q}\right)
\right\} ,  \label{Static_CMR_Robust_representation}
\end{equation}%
where 
\begin{equation}
\psi _{\rho }^{\ast }\left( \mathbb{Q}\right) :=\sup\limits_{X\in \mathcal{A}%
\rho }\mathbb{E}_{\mathbb{Q}}\left[ -X\right] ,  \label{Def._minimal_penalty}
\end{equation}%
and $\mathcal{A}_{\rho }:=\left\{ X\in \mathfrak{M}_{b}:\rho (X)\leq
0\right\} $ is the \textit{acceptance set} of $\rho .$

The penalty $\psi _{\rho }^{\ast }$ is called the \textit{minimal penalty
function} associated to $\rho $ because, for any other penalty function $%
\psi $ fulfilling $\left( \ref{Static_CMR_Robust_representation}\right) $, $%
\psi \left( \mathbb{Q}\right) \geq \psi _{\rho }^{\ast }\left( \mathbb{Q}%
\right) $, for all $\mathbb{Q}\in \mathcal{Q}_{cont}.$ Furthermore, for the
minimal penalty function, the next biduality relation is satisfied 
\begin{equation}
\psi _{\rho }^{\ast }\left( \mathbb{Q}\right) =\sup_{X\in \mathfrak{M}%
_{b}\left( \Omega ,\mathcal{F}\right) }\left\{ \mathbb{E}_{\mathbb{Q}}\left[
-X\right] -\rho \left( X\right) \right\} ,\quad \forall \mathbb{Q\in }%
\mathcal{Q}_{cont}.  \label{static convex rsk msr biduality}
\end{equation}

\begin{remark}
Among the measures of risk, the class of them that are concentrated on the
set of probability measures $\mathcal{Q\subset Q}_{cont}$ are of special
interest. Recall that a function $I:E\subset \mathbb{R}^{\Omega }\rightarrow 
\mathbb{R}$ is \textit{sequentially continuous from below (above)} when $%
\left\{ X_{n}\right\} _{n\in \mathbb{N}}\uparrow X\Rightarrow
\lim_{n\rightarrow \infty }I\left( X_{n}\right) =I\left( X\right) $ (
respectively $\left\{ X_{n}\right\} _{n\in \mathbb{N}}\downarrow
X\Rightarrow \lim_{n\rightarrow \infty }I\left( X_{n}\right) =I\left(
X\right) $). F\"{o}llmer and Schied \cite{FoellSch 2002 a} proved that any
sequentially continuous from below convex measure of risk is concentrated on
the set $\mathcal{Q}$. Later, Kr\"{a}tschmer \cite[Prop. 3 p. 601]%
{Kraetschmer 2005} established that the sequential continuity from below is
not only a sufficient but also a necessary condition in order to have a
representation, by means of the minimal penalty function in terms of
probability measures.
\end{remark}

\section{The market model \label{Sect. Market_Model}}

\setcounter{equation}{0} In this section, we introduce the market model
considered in this paper. It is based on the generalization of the classical
geometric Brownian setting, but in this case the coefficients are not
constant and jumps are included in the model through an exogenous stochastic
process. One of the most debatable feature about an stochastic process used
for modelling stock market prices is the issue about the independent
increments. A remarkable property of the proposed model is the fact that it
need not to have independent increments. Also, it includes a certain
subclase of exponential L\'{e}vy models. This section is concluded with a
characterization of the set of equivalent local martingale measures.

\subsection{General description and martingale measures}

First, consider the stochastic process $Y_{t}$ with dynamics given by 
\begin{equation}
Y_{t}:=\int\limits_{]0,t]}\alpha _{s}ds+\int\limits_{]0,t]}\beta
_{s}dW_{s}+\int\limits_{]0,t]\times \mathbb{R}_{0}}\gamma \left( s,x\right)
\left( \mu \left( ds,dx\right) -\nu \left( dx\right) \ ds\right) ,
\label{Def._Y}
\end{equation}%
where the processes $\alpha ,\beta $ are c\'{a}dl\'{a}g, with $\beta \in 
\mathcal{L}\left( W\right) $ and $\gamma \in \mathcal{G}\left( \mu \right) $
. Throughout we assume that the coefficients $\alpha $, $\beta $ and $\gamma 
$ fulfill the following conditions: 
\begin{equation*}
\begin{array}{cl}
\left( A\ 1\right) & \int\nolimits_{]0,t]}\left( \alpha _{s}\right)
^{2}ds<\infty \quad \forall t\in \mathbb{R}_{+}\quad \mathbb{P}\text{-}a.s.\
. \\ 
\smallskip \  &  \\ 
\left( A\ 2\right) & 0<c\leq \left\vert \beta _{t}\right\vert \quad \forall
t\in \mathbb{R}_{+}\quad \mathbb{P}\text{-}a.s.\ . \\ 
\smallskip \  &  \\ 
\left( A\ 3\right) & \int_{0}^{T}\left( \frac{\alpha _{u}}{\beta _{u}}%
\right) ^{2}du\in \mathcal{L}^{\infty }\left( \mathbb{P}\right) . \\ 
\smallskip \  &  \\ 
\left( A\ 4\right) & \gamma \left( t,\bigtriangleup L_{t}\right) \times 
\mathbf{1}_{\mathbb{R}_{0}}\left( \bigtriangleup L_{t}\right) \geq -1\quad
\forall t\in \mathbb{R}_{+}\quad \mathbb{P}\text{-}a.s.\ . \\ 
\smallskip \  &  \\ 
\left( A\ 5\right) & \left\{ \gamma \left( t,\bigtriangleup L_{t}\right) 
\mathbf{1}_{\mathbb{R}_{0}}\left( \bigtriangleup L_{t}\right) \right\}
_{t\in \mathbb{R}_{+}}\text{ is a locally bounded process.}%
\end{array}%
\end{equation*}

The market model consists of two assets, one of them is the num\'eraire,
having a strictly positive price. The dynamics of the other risky asset will
be modeled as a function of the process $Y_t$ defined above. More
specifically, since we are interested in the problem of robust utility
maximization, the discounted capital process can be written in terms of the
wealth invested in this asset, and hence the problem can be written using
only the dynamics of the discounted price of this asset. For this reason,
throughout we will be concentrated in the dynamics of this price.

The dynamics of the discounted price process $S$ are determined by the
process $Y$ as its Doleans-Dade exponential 
\begin{equation}
S_{t}=S_{0}\mathcal{E}\left( Y_{t}\right) .  \label{Model_1_S_Def.}
\end{equation}%
Condition $\left( A\ 4\right) $ ensures that the price process is
non-negative. This process is an exponential semimartingale, as it would be
the case of an arbitrary semimartingale $Y$, if and only if the following
two conditions are fulfilled: 
\begin{equation}
\begin{array}{cl}
\left( i\right) & S=S\mathbf{1}_{\left[ 0,\tau \right] }\text{, for }\tau
:=\inf \left\{ t>0:S_{t}=0\text{ or }S_{t-}=0\right\} . \\ 
&  \\ 
\left( ii\right) & \frac{1}{S_{t-}}\mathbf{1}_{\left[ S_{t-}\neq 0\right] }%
\text{ is integrable w.r.t. }S.%
\end{array}
\label{exponential semimartingale iff}
\end{equation}%
The first property is conceptually very appropriate when we are interested
in modelling the dynamics of a price process. Recall that a stochastically
continuous semimartingale has independent increments if and only if its
predictable triplet is non-random. Therefore, in general, the price process $%
S$ is not a L\'{e}vy exponential model, because $\left[ Y^{c}\right]
_{t}=\int\nolimits_{0}^{t}\left( \beta _{u}\right) ^{2}du$ need not to be
deterministic. However, observe that the price dynamics $\left( \ref%
{Model_1_S_Def.}\right) $ includes L\'{e}vy exponential models, for L\'{e}vy
processes with $\bigtriangleup L_{t}\geq -1.$

For the model $\left( \ref{Model_1_S_Def.}\right) $ the price process can be
written explicitly as%
\begin{equation}
\begin{array}{rl}
S_{t}= & S_{0}\exp \left\{ \int\limits_{]0,t]}\alpha
_{s}ds+\int\limits_{]0,t]}\beta _{s}dW_{s}+\int\limits_{]0,t]\times \mathbb{R%
}_{0}}\gamma \left( s,x\right) \left( \mu \left( ds,dx\right) -\nu \left(
dx\right) \ ds\right) -\frac{1}{2}\int\limits_{]0,t]}\left( \beta
_{s}\right) ^{2}ds\right\} \\ 
& \times \exp \left\{ \int\limits_{]0,t]\times \mathbb{R}_{0}}\{\ln \left(
1+\gamma \left( s,x\right) \right) -\gamma \left( s,x\right) \}\mu \left(
ds,dx\right) \right\}.%
\end{array}
\label{Model_1_S(t)_Explicit}
\end{equation}%
Observe that $\left( A\ 5\right) $ is a necessary and sufficient condition
for $S$ to be a locally bounded process.

The predictable c\'adl\'ag process $\left\{ \pi _{t}\right\} _{t\in \mathbb{R%
}_{+}}$, satisfying the integrability condition $\int_{0}^{t}\left( \pi
_{s}\right) ^{2}ds<\infty $ $\mathbb{P}$-a.s. for all $t\in \mathbb{R}_{+}$,
shall denote the proportion of wealth at time $t$ invested in the risky
asset $S$. For an initial capital $x $, the discounted wealth $X_{t}^{x,\pi
} $ associated with a self-financing investment strategy $\left( x,\pi
\right) $ fulfills the equation%
\begin{equation}
X_{t}^{x,\pi }=x+\int_{0}^{t}\frac{X_{u-}^{x,\pi }\pi _{u}}{S_{u-}}\mathbf{1}%
_{\left[ S_{u-}\neq 0\right] }dS_{u}.  \label{Model_1_Def._X}
\end{equation}

We say that a self-financing strategy $\left( x,\pi \right) $ is \textit{%
admissible} if the wealth process satisfies $X_{t}^{x,\pi }>0$ for all $t>0$%
. The class of admissible wealth processes with initial wealth less than or
equal to $x$ is denoted by $\mathcal{X}\left( x\right) .$

Next result characterizes the class of \textit{\ equivalent local martingale
measures} defined as 
\begin{equation}
\mathcal{Q}_{elmm}:=\{\mathbb{Q}\in \mathcal{Q}_{\approx }(\mathbb{P}):%
\mathcal{X}\left( 1\right) \subset \mathcal{M}_{loc}\left( \mathbb{Q}\right)
\}=\{\mathbb{Q}\in \mathcal{Q}_{\approx }(\mathbb{P}):S\in \mathcal{M}%
_{loc}\left( \mathbb{Q}\right) \}.  \label{Def._Qelmm(P)}
\end{equation}%
The class of density processes associated with $\mathcal{Q}_{elmm}$ is
denoted by $\mathcal{D}_{elmm}\left( \mathbb{P}\right) .$ Kunita \cite%
{Kunita 2004} gave conditions on the parameters $\left( \theta _{0},\theta
_{1}\right) $ of a measure $\mathbb{Q}\in \mathcal{Q}_{\approx }$ in order
that it is a local martingale measure for a L\'{e}vy exponential model i.e.
when $S=\mathcal{E}\left( L\right) $. Observe that in this case $\mathcal{Q}%
_{elmm}\left( S\right) =\mathcal{Q}_{elmm}\left( L\right) .$ Next
proposition extends those results, giving conditions on the parameters $%
\left( \theta _{0},\theta _{1}\right) $ under which an equivalent measure is
a local martingale measure for the price model (\ref{Model_1_S_Def.}).

\begin{proposition}
\label{Q ELMM iff} Given $\mathbb{Q}\in \mathcal{Q}_{\approx }$, let $\theta
_{0}\in \mathcal{L}\left( W\right) $ and $\theta _{1}\in \mathcal{G}\left(
\mu \right) $ be the corresponding processes describing the density
processes found in Lemma \ref{Q<<P =>}. Then, the following equivalence
holds: 
\begin{equation}
\mathbb{Q}\in \mathcal{Q}_{elmm}\Longleftrightarrow \alpha _{t}+\beta
_{t}\theta _{0}\left( t\right) +\int\nolimits_{\mathbb{R}_{0}}\gamma \left(
t,x\right) \theta _{1}\left( t,x\right) \nu \left( dx\right) =0\ \forall
t\geq 0\ \quad \mathbb{P}\text{-a.s. }  \label{Qelmm <=>}
\end{equation}
\end{proposition}

\begin{proof}
Let $\mathbb{Q}\in \mathcal{Q}_{\approx }$ be an equivalent probability
measure with density process given by $D_{t}:=\mathbb{E}\left[ \left. d%
\mathbb{Q}/d\mathbb{P}\right\vert \mathcal{F}_{t}\right] =\mathcal{E}\left(
Z^{\theta }\right) _{t}$, where the last equality follows from Lemma \ref%
{Q<<P =>}. Then, we have that 
\begin{equation*}
S\in \mathcal{M}_{loc}^{1}\left( \mathbb{Q}\right) \Longleftrightarrow SD\in 
\mathcal{M}_{loc}^{1}\left( \mathbb{P}\right) .
\end{equation*}

Since $\theta _{1},\gamma \in \mathcal{G}\left( \mu \right) ,$ from $\left(
A\ 5\right) $ the process $\left\{ \gamma \left( t,\bigtriangleup
L_{t}\right) \mathbf{1}_{\mathbb{R}_{0}}\left( \bigtriangleup L_{t}\right)
\right\} _{t\in \mathbb{R}_{+}}$ is a locally bounded process, we have that $%
\int \gamma \theta _{1}d\mu \in \mathcal{A}_{loc}$ , which yields that $%
\gamma \theta _{1}\in \mathcal{G}\left( \mu \right) $ and 
\begin{equation*}
\int \gamma \theta _{1}d\left\{ \mu -\mu ^{\mathcal{P}}\right\} =\int \gamma
\theta _{1}d\mu -\int \gamma \theta _{1}d\mu ^{\mathcal{P}}.
\end{equation*}%
Therefore, 
\begin{equation*}
\left[ Y,Z^{\theta }\right] _{t}=\int\limits_{0}^{t}\beta _{s}\theta
_{0}ds+\int\limits_{]0,t]\times \mathbb{R}_{0}}\gamma \theta _{1}d\left\{
\mu -\mu ^{\mathcal{P}}\right\} +\int\limits_{]0,t]\times \mathbb{R}%
_{0}}\gamma \theta _{1}d\mu ^{\mathcal{P}}.
\end{equation*}%
Now, we write 
\begin{equation*}
S_{t}D_{t}=S_{0}\mathcal{E}\left( Y\right) _{t}\mathcal{E}\left( Z^{\theta
}\right) _{t}=S_{0}\mathcal{E}\left( Y+Z^{\theta }+\left[ Y,Z^{\theta }%
\right] \right) _{t},
\end{equation*}%
and making some rearrangements we have that 
\begin{eqnarray*}
&&S_{t}D_{t} \\
&=&S_{0}+\int S_{u-}D_{u-}d\left\{ Y+Z^{\theta }+\left[ Y,Z^{\theta }\right]
\right\} _{u} \\
&=&S_{0}+\int S_{u-}D_{u-}d\left\{ \int \left( \beta +\theta _{0}\right)
dW+\int \left( \gamma +\theta _{1}+\gamma \theta _{1}\right) d\left\{ \mu
-\mu ^{\mathcal{P}}\right\} \right\} _{u} \\
&&+\int S_{u-}D_{u-}d\left\{ \int \left( \alpha _{s}+\beta _{s}\theta
_{0}\left( s\right) +\int \gamma \theta _{1}\nu \left( dx\right) \right)
ds\right\} _{u}.
\end{eqnarray*}%
On the other hand, observe that 
\begin{equation*}
\int S_{u-}D_{u-}d\left\{ \int \left( \beta +\theta _{0}\right) dW+\int
\left( \gamma +\theta _{1}+\gamma \theta _{1}\right) d\left\{ \mu -\mu ^{%
\mathcal{P}}\right\} \right\} _{u}
\end{equation*}%
belongs to the set of local martingales $\mathcal{M}_{loc},$ and 
\begin{equation*}
\int S_{u-}D_{u-}d\left\{ \int \left( \alpha _{s}+\beta _{s}\theta
_{0}\left( s\right) +\int \gamma \theta _{1}\nu \left( dx\right) \right)
ds\right\}
\end{equation*}%
is a finite variation continuous process in $\mathcal{V}^{c}$. To verify
this claim, observe first that $\left( A\ 1\right) $ implies that $%
\int\nolimits_{0}^{t}\alpha _{s}ds\in $ $\mathcal{V}.$ Further, for $\beta
_{s},\theta _{0}\in \mathcal{L}\left( W\right) $ we know that $%
\int\nolimits_{\left[ 0,t\right] }\left\{ \beta _{s}\right\} ^{2}ds<\infty $ 
$\mathbb{P}$-a.s., and $\int\nolimits_{\left[ 0,t\right] }\left\{ \theta
_{0}\left( s\right) \right\} ^{2}ds<\infty $ $\mathbb{P}$-a.s., and from the
Rogers-H\"{o}lder inequality%
\begin{equation*}
\int_{0}^{t}\left\vert \beta _{s}\right\vert \left\vert \theta _{0}\left(
s\right) \right\vert ds\leq \left( \int_{0}^{t}\left( \beta _{s}\right)
^{2}ds\right) ^{\frac{1}{2}}\left( \int_{0}^{t}\left( \theta _{0}\left(
s\right) \right) ^{2}ds\right) ^{\frac{1}{2}}<\infty .
\end{equation*}%
Then, $\int\limits_{0}^{t}\beta _{s}\theta _{0}ds$ is of finite variation
due to the absolutely integrability of the integrand, i.e. $%
\int\nolimits_{0}^{t}\beta _{s}\theta _{0}ds\in $ $\mathcal{V}$. Since $\int
\gamma \left( s,x\right) \theta _{1}\left( s,x\right) \mu \left(
ds,dx\right) \in \mathcal{A}_{loc}$, it follows that 
\begin{equation*}
\int\nolimits_{\left[ 0,t\right] \times \mathbb{R}_{0}}\gamma \left(
s,x\right) \theta _{1}\left( s,x\right) \nu \left( dx\right) ds\in \mathcal{%
V\ }\mathbb{P}-a.s.\forall t\in \mathbb{R}_{+}.
\end{equation*}%
Summarizing, 
\begin{equation*}
\int\nolimits_{0}^{t}\alpha _{s}ds+\int\nolimits_{0}^{t}\beta _{s}\theta
_{0}ds+\int\limits_{]0,t]\times \mathbb{R}_{0}}\gamma \theta _{1}\nu \left(
dx\right) ds\in \mathcal{V}.
\end{equation*}%
The equivalence (\ref{Qelmm <=>}) follows now observing that a predictable
local martingale with locally integrable variation is constant.
\end{proof}

\subsection{Minimal penalties}

\label{MinimalPenalties}

Now, we shall introduce a family of penalty functions for the density
processes described in Section \ref{Sub_Sect:_Density_processes}, for the
absolutely continuous measures $\mathbb{Q}\in \mathcal{Q}_{\ll }\left( 
\mathbb{P}\right) $.

Let $h_{0}$\thinspace and $h_{1}$ be $\mathbb{R}_{+}$-valued convex
functions defined in $\mathbb{R}$ with $h_{0}\left( 0\right) =0=h_{1}\left(
0\right) $, and $h:\mathbb{R}_{+}\rightarrow \mathbb{R}_{+}$ be increasing
convex function continuous at zero with $h\left( 0\right) =0$. Define the
penalty function 
\begin{equation}
\begin{array}{rl}
\vartheta \left( \mathbb{Q}\right) := & \mathbb{E}_{\mathbb{Q}}\left[
\int\limits_{0}^{T}h\left( h_{0}\left( \theta _{0}\left( t\right) \right)
+\int\nolimits_{\mathbb{R}_{0}}\delta \left( t,x\right) h_{1}\left( \theta
_{1}\left( t,x\right) \right) \nu \left( dx\right) \right) dt\right] \mathbf{%
1}_{\mathcal{Q}_{\ll }}\left( \mathbb{Q}\right) \\ 
& +\infty \times \mathbf{1}_{\mathcal{Q}_{cont}\setminus \mathcal{Q}_{\ll
}}\left( \mathbb{Q}\right) ,%
\end{array}
\label{Def._penalty_theta}
\end{equation}%
where $\theta _{0},$ $\theta _{1}$ are the processes associated to $\mathbb{Q%
}$ from Lemma \ref{Q<<P =>} and $\delta \left( t,x\right) :\mathbb{R}%
_{+}\times \mathbb{R}_{0}\rightarrow \mathbb{R}_{+}$ is an arbitrary but fix
nonnegative function $\delta \left( t,x\right) \in \mathcal{G}\left( \mu
\right) $. Further, define the convex measure of risk 
\begin{equation}
\rho \left( X\right) :=\sup_{\mathbb{Q\in }\mathcal{Q}_{\ll }(\mathbb{P}%
)}\left\{ \mathbb{E}_{\mathbb{Q}}\left[ -X\right] -\vartheta \left( \mathbb{Q%
}\right) \right\} .  \label{rho def.}
\end{equation}%
Notice that $\rho $ is a normalized and sensitive measure of risk . Next
theorem establishes the minimality of the penalty function introduced above
for the risk measure $\rho $. The proof can be found in \cite{HH-PH}.

\begin{theorem}
\label{theta=minimal penalty function} The penalty function $\vartheta $
defined in $\left( \ref{Def._penalty_theta}\right) $ is the minimal penalty
function of the convex risk measure $\rho $ given by $\left( \ref{rho def.}%
\right) $.
\end{theorem}

\section{Robust utility maximization \label{Sect Robust Utility Maximization}%
}

\setcounter{equation}{0} In this section the connection between penalty
functions and the existence of solutions to the penalized robust expected
utility problem is established. We also formulate the dual problem in terms
of control processes for an arbitrary utility function.

\subsection{Penalties and solvability}

Let us now introduce the class 
\begin{equation}
\mathcal{C}:=\left\{ \mathcal{E}\left( Z^{\xi }\right) :%
\begin{array}{l}
\xi :=\left( \xi ^{\left( 0\right) },\xi ^{\left( 1\right) }\right) ,\ \xi
^{\left( 0\right) }\in \mathcal{L}\left( W\right) ,\ \xi ^{\left( 1\right)
}\in \mathcal{G}\left( \mu \right) ,\text{ with} \\ 
\alpha _{t}+\beta _{t}\xi _{t}^{\left( 0\right) }+\int\limits_{\mathbb{R}%
_{0}}\gamma \left( t,x\right) \xi ^{\left( 1\right) }\left( t,x\right) \nu
\left( dx\right) =0\ \text{Lebesgue }\forall t%
\end{array}%
\ \right\} ,  \label{Def._C=Control_Set}
\end{equation}%
with $Z^{\xi }$ as in (\ref{Def._Ztheta(t)}). Observe that $\mathcal{D}%
_{elmm}\left( \mathbb{P}\right) \subset \mathcal{C}\subset \mathcal{Y}_{%
\mathbb{P}}\left( 1\right) $; see (\ref{Def._Yq(y)}) for the definition of $%
\mathcal{Y}_{\mathbb{P}}\left( 1\right) $. It should be pointed out that
this relation between these three sets plays a crucial role in the
formulation of the dual problem, even in the non-robust case.

\begin{theorem}
\label{Power_utility:_Existence} For $q\in \left( -\infty ,1\right)
\backslash \left\{ 0\right\} $, let $U\left( x\right) :=\frac{1}{q}x^{q}$ be
the power utility function, and consider the functions $h,h_{0}$ and $h_{1}$
as in Subsection \ref{MinimalPenalties}, satisfying the following conditions:%
\begin{equation*}
\begin{array}{l}
h\left( x\right) \geq \exp \left( \kappa _{1}x^{2}\right) -1\;\text{\textrm{%
where }}\kappa _{1}:=1\vee 2\left( 2p^{2}+p\right) T\;\text{\textrm{and }}%
\;p:=\frac{q}{1-q}, \\ 
h_{0}\left( x\right) \geq \left\vert x\right\vert , \\ 
h_{1}\left( x\right) \geq \frac{\left\vert x\right\vert }{c},\text{\textrm{%
for }}c\text{ \textrm{as in assumption }}\left( A\ 2\right) .%
\end{array}%
\end{equation*}%
Then, for the penalty function 
\begin{equation*}
\vartheta _{x^{q}}\left( \mathbb{Q}\right) :=\mathbb{E}_{\mathbb{Q}}\left[
\int\limits_{0}^{T}h\left( h_{0}\left( \theta _{0}\left( t\right) \right)
+\int\nolimits_{\mathbb{R}_{0}}\left\vert \gamma \left( t,x\right)
\right\vert h_{1}\left( \theta _{1}\left( t,x\right) \right) \nu \left(
dx\right) \right) dt\right] ,
\end{equation*}%
the penalized robust utility maximization problem $\left( \ref%
{robust_probl._primal_value_funct}\right) $ has a solution.
\end{theorem}

\begin{proof}
The penalty function $\vartheta _{x^{q}}$ is bounded from below, and by
Theorem \ref{theta=minimal penalty function} it is the minimal penalty
function of the normalized and sensitive convex measure of risk defined in $%
\left( \ref{Static_CMR_induced_by_phi}\right) $. Therefore, we only need to
prove that condition $\left( \ref{vQ(y)<oo_for_y>0_Q_nice}\right) $ holds.
In order to prove that, fix an arbitrary probability measure $\mathbb{Q}\in 
\mathcal{Q}_{\approx }^{\vartheta _{x^{q}}}=\left\{ \mathbb{Q}\approx 
\mathbb{P}:\vartheta _{x^{q}}\left( \mathbb{Q}\right) <\infty \right\} $ and
let $\theta =\left( \theta _{0},\theta _{1}\right) $ be the corresponding
coefficients obtained in Lemma \ref{Q<<P =>}.

$\left( 1\right) $ In Lemma 4.2, Schied \cite{Schd 2007} establishes that
even for $\mathbb{Q}\in \mathcal{Q}_{\ll }$, with density process $D$, the
next equivalence holds 
\begin{equation*}
Y\in \mathcal{Y}_{\mathbb{Q}}\left( y\right) \Leftrightarrow YD\in \mathcal{Y%
}_{\mathbb{P}}\left( y\right) .
\end{equation*}%
Therefore, for $\mathbb{Q}\in \mathcal{Q}_{\ll }^{\vartheta _{x^{q}}}$, with
coefficient $\theta =\left( \theta _{0},\theta _{1}\right) $, it follows
that 
\begin{equation*}
\begin{array}{cll}
v_{\mathbb{Q}}\left( y\right) & \equiv \inf_{Y\in \mathcal{Y}_{\mathbb{Q}%
}\left( y\right) }\left\{ \mathbb{E}_{\mathbb{Q}}\left[ V\left( Y_{T}\right) %
\right] \right\} &  \\ 
& =\inf_{Y\in \mathcal{Y}_{\mathbb{P}}\left( 1\right) }\left\{ \mathbb{E}_{%
\mathbb{Q}}\left[ V\left( y\frac{Y_{T}}{D_{T}^{\mathbb{Q}}}\right) \right]
\right\} & \leq \inf_{\xi \in \mathcal{C}}\left\{ \mathbb{E}_{\mathbb{Q}}%
\left[ V\left( y\frac{\mathcal{E}\left( Z^{\xi }\right) _{T}}{\mathcal{E}%
\left( Z^{\theta }\right) _{T}}\right) \right] \right\} .%
\end{array}%
\end{equation*}

$\left( 2\right) $ Define 
\begin{equation*}
\varepsilon _{t}:=\alpha _{t}+\beta _{t}\theta _{0}\left( t\right)
+\int\limits_{\mathbb{R}_{0}}\gamma \left( t,x\right) \theta _{1}\left(
t,x\right) \nu \left( dx\right),
\end{equation*}%
the process involved in the definition of the class $\mathcal{C}$ in (\ref%
{Def._C=Control_Set}).

When $\varepsilon _{t}$ is identically zero for all $t>0$, Proposition \ref%
{Q ELMM iff} implies that $\mathbb{Q}\in \mathcal{Q}_{elmm}.$ However, for $%
\mathbb{Q}\in \mathcal{Q}_{elmm}$ the constant process $Y\equiv y$ belongs
to $\mathcal{Y}_{\mathbb{Q}}\left( y\right) $, and it follows that $v_{%
\mathbb{Q}}\left( y\right) <\infty ,$ for all $y>0$. In this case the proof
is concluded.

If $\varepsilon $ is not identically zero, consider $\xi _{t}^{\left(
0\right) }:=\theta _{0}\left( t\right) -\frac{\varepsilon _{t}}{\beta _{t}}$
and $\xi ^{\left( 1\right) }:=\theta _{1}.$ Since%
\begin{equation*}
\infty >\vartheta _{x^{q}}\left( \mathbb{Q}\right) \geq \mathbb{E}_{\mathbb{Q%
}}\left[ \int\limits_{0}^{T}\left( \frac{1}{\beta _{t}}\int\nolimits_{%
\mathbb{R}_{0}}\gamma \left( t,x\right) \theta _{1}\left( t,x\right) \nu
\left( dx\right) \right) ^{2}dt\right] -T,
\end{equation*}%
it follows that $\left\{ \frac{1}{\beta _{t}}\int\nolimits_{\mathbb{R}%
_{0}}\gamma \left( t,x\right) \theta _{1}\left( t,x\right) \nu \left(
dx\right) \right\} _{t\in \left[ 0,T\right] }\in \mathcal{L}\left( W^{\prime
}\right) $ for $W^{\prime }$ a $\mathbb{Q}$-Wiener process and thus also $%
\xi ^{\left( 0\right) }\in \mathcal{L}\left( W^{\prime }\right) .$ Moreover,
for $\xi =\left( \xi ^{\left( 0\right) },\xi ^{\left( 1\right) }\right) $ we
have that $\mathcal{E}\left( Z^{\xi }\right) \in \mathcal{C}$.

Using Girsanov's theorem, we obtain $\frac{\mathcal{E}\left( Z^{\xi }\right)
_{t}}{\mathcal{E}\left( Z^{\theta }\right) _{t}}=\exp \left\{
\int\nolimits_{]0,t]}\left( -\frac{\varepsilon _{u}}{\beta _{u}}\right)
dW_{u}^{\prime }-\frac{1}{2}\int\nolimits_{]0,t]}\left( \frac{\varepsilon
_{u}}{\beta _{u}}\right) ^{2}du\right\} $.

$\left( 3\right) $ The Cauchy-Bunyakovsky-Schwarz inequality yields 
\begin{equation}
\begin{array}{rl}
\mathbb{E}_{\mathbb{Q}}\left[ V\left( y\frac{\mathcal{E}\left( Z^{\xi
}\right) _{T}}{\mathcal{E}\left( Z^{\theta }\right) _{T}}\right) \right] = & 
\frac{1}{p}y^{-p}\mathbb{E}_{\mathbb{Q}}\left[ \exp \left\{
p\int\limits_{]0,T]}\left( \frac{\varepsilon _{t}}{\beta _{t}}\right)
dW^{\prime }+\frac{p}{2}\int\limits_{]0,T]}\left( \frac{\varepsilon _{t}}{%
\beta _{t}}\right) ^{2}dt\right\} \right] \\ 
\leq & \frac{1}{p}y^{-p}\mathbb{E}_{\mathbb{Q}}\left[ \exp \left\{
2p\int\limits_{]0,T]}\left( \frac{\varepsilon _{t}}{\beta _{t}}\right)
dW^{\prime }-\frac{4p^{2}}{2}\int\limits_{]0,T]}\left( \frac{\varepsilon _{t}%
}{\beta _{t}}\right) ^{2}dt\right\} \right] ^{\frac{1}{2}} \\ 
& \times \mathbb{E}_{\mathbb{Q}}\left[ \exp \left\{ \left( \frac{4p^{2}}{2}%
+p\right) \int\limits_{]0,T]}\left( \frac{\varepsilon _{t}}{\beta _{t}}%
\right) ^{2}dt\right\} \right] ^{\frac{1}{2}}.%
\end{array}
\label{FiniteV}
\end{equation}

On the other hand, the process 
\begin{equation*}
\exp \left\{ 2p\int\limits_{]0,T]}\left( \frac{\varepsilon _{t}}{\beta _{t}}%
\right) dW^{\prime }-\frac{4p^{2}}{2}\int\limits_{]0,T]}\left( \frac{%
\varepsilon _{t}}{\beta _{t}}\right) ^{2}dt\right\} \in \mathcal{M}%
_{loc}\left( \mathbb{Q}\right)
\end{equation*}%
is a local $\mathbb{Q}$-martingale and, since it is positive, is a
supermartingale. Hence, 
\begin{equation*}
\mathbb{E}_{\mathbb{Q}}\left[ \exp \left\{ 2p\int\limits_{]0,T]}\left( \frac{%
\varepsilon _{t}}{\beta _{t}}\right) dW^{\prime }-\frac{4p^{2}}{2}%
\int\limits_{]0,T]}\left( \frac{\varepsilon _{t}}{\beta _{t}}\right)
^{2}dt\right\} \right] \leq 1.
\end{equation*}%
Finally, observe that for $\mathbb{Q}\in \mathcal{Q}_{\ll }^{\vartheta
_{x^{q}}}$, using that it has finite penalization $\vartheta _{x^{q}}\left( 
\mathbb{Q}\right) <\infty $ and Jensen's inequality, we have 
\begin{eqnarray*}
\infty &>&\mathbb{E}_{\mathbb{Q}}\left[ \exp \left\{ \frac{\kappa _{1}}{T}%
\int\limits_{0}^{T}\left( h_{0}\left( \theta _{0}\left( t\right) \right)
+\int\limits_{\mathbb{R}_{0}}\left\vert \gamma \left( t,x\right) \right\vert
h_{1}\left( \theta _{1}\left( t,x\right) \right) \nu \left( dx\right)
\right) ^{2}dt\right\} \right] \\
&\geq &\mathbb{E}_{\mathbb{Q}}\left[ \exp \left\{ 2\left( 2p^{2}+p\right)
\int\limits_{0}^{T}\left( \left\vert \theta _{0}\left( t\right) \right\vert +%
\frac{1}{\left\vert \beta _{t}\right\vert }\left\vert \int\nolimits_{\mathbb{%
R}_{0}}\gamma \left( t,x\right) \theta _{1}\left( t,x\right) \nu \left(
dx\right) \right\vert \right) ^{2}dt\right\} \right] .
\end{eqnarray*}%
From the last two displays it follows that the r.h.s. of (\ref{FiniteV}) is
finite and the theorem follows.
\end{proof}

Next theorem establishes a sufficient condition for the existence of a
solution to the robust utility maximization problem $\left( \ref%
{robust_probl._primal_value_funct}\right) $ for an arbitrary utility
function.

\begin{theorem}
\label{General_Utility:_Existence} Suppose that the utility function $%
\widetilde{U}$ is bounded above by a power utility $U$, with penalty
function $\vartheta _{x^{q}}$ associated to $U$ as in Theorem \ref%
{Power_utility:_Existence}. Then, the robust utility maximization problem $%
\left( \ref{robust_probl._primal_value_funct}\right) $ for $\widetilde{U}$
with penalty $\vartheta _{x^{q}}$ has an optimal solution.
\end{theorem}

\begin{proof}
Since $U\left( x\right) :=\frac{1}{q}x^{-q}\geq \widetilde{U}\left( x\right) 
$ for all $x>0$, for some $q\in \left( -\infty ,1\right) \backslash \left\{
0\right\} $ the corresponding convex conjugate functions satisfy $V\left(
y\right) \geq \widetilde{V}\left( y\right) $ for each $y>0.$ As it was
pointed out in Remark \ref{V<0=>Solution exist}, we can restrict ourself to
the positive part $\widetilde{V}^{+}\left( y\right) .$ From Proposition \ref%
{Power_utility:_Existence}, we can fix some $Y\in \mathcal{Y}_{\mathbb{Q}%
}\left( y\right) $ such that $\mathbb{E}_{\mathbb{Q}}\left[ V\left(
Y_{T}\right) \right] <\infty $ for any $\mathbb{Q\in }\mathcal{Q}_{\approx
}^{\vartheta _{x^{q}}}$ and $y>0$, arbitrary, but fixed. Furthermore, the
inequality $V\left( y\right) \geq \widetilde{V}\left( y\right) $ implies
that their inverse functions satisfy $\left( V^{+}\right) ^{\left( -1\right)
}\left( n\right) \geq \left( \widetilde{V}^{+}\right) ^{\left( -1\right)
}\left( n\right) $ for all $n\in \mathbb{N}$, and hence 
\begin{equation*}
\sum_{n=1}^{\infty }\mathbb{Q}\left[ Y_{T}\leq \left( \widetilde{V}%
^{+}\right) ^{\left( -1\right) }\left( n\right) \right] \leq
\sum_{n=1}^{\infty }\mathbb{Q}\left[ Y_{T}\leq \left( V^{+}\right) ^{\left(
-1\right) }\left( n\right) \right] <\infty .
\end{equation*}%
The moments Lemma ($\mathbb{E}_{\mathbb{Q}}\left[ X\right] <\infty
\Leftrightarrow \sum_{n=1}^{\infty }\mathbb{Q}\left[ \left\vert X\right\vert
\geq n\right] <\infty $) yields $\mathbb{E}_{\mathbb{Q}}\left[ \widetilde{V}%
^{+}\left( Y_{T}\right) \right] <\infty $, and the assertion follows.
\end{proof}

\begin{example}
\label{Log-utility:_dominated_by_power}The logarithm utility function
satisfies conditions of Theorem \ref{General_Utility:_Existence}. However,
this case will be studied more deeply in Section \ref{Sect:_log_utility_case}%
, since the techniques involve interesting arguments related to the relative
entropy.
\end{example}

From the proof of Theorem \ref{General_Utility:_Existence} it is clear that
the behavior of the convex conjugate function in a neighborhood of zero is
fundamental. From this observation we conclude the following.

\begin{corollary}
Let $U$ be a utility function with convex conjugate $V$, and $\vartheta $ a
penalization function such that the robust utility maximization problem (\ref%
{robust_probl._primal_value_funct}) has a solution. For a utility function $%
\widetilde{U}$ such that their convex conjugate function $\widetilde{V}$ is
majorized in an $\varepsilon $-neighborhood of zero by $V$, the
corresponding utility maximization problem (\ref%
{robust_probl._primal_value_funct}) has a solution.
\end{corollary}

Next we give an alternative representation of the robust dual value
function, introduced in (\ref{robust_probl_dual_value_funct}), in terms of
the family $\mathcal{C}$ of stochastic processes.

\begin{theorem}
\label{Model_1_robust dual value function_(Thm)}For a utility function $U$
satisfying condition $\left( \ref{vQ(y)<oo_for_y>0_Q_nice}\right) $, the
dual value function can be written as 
\begin{equation}
\begin{array}{cl}
v\left( y\right) & =\inf\limits_{\mathbb{Q}\in \mathcal{Q}_{\approx
}^{\vartheta }}\left\{ \inf_{\xi \in \mathcal{C}}\left\{ \mathbb{E}_{\mathbb{%
Q}}\left[ V\left( y\tfrac{\mathcal{E}\left( Z^{\xi }\right) _{T}}{D_{T}^{%
\mathbb{Q}}}\right) \right] \right\} +\vartheta \left( \mathbb{Q}\right)
\right\} \\ 
& =\inf\limits_{\mathbb{Q}\in \mathcal{Q}_{\ll }}\left\{ \inf_{\xi \in 
\mathcal{C}}\left\{ \mathbb{E}_{\mathbb{Q}}\left[ V\left( y\tfrac{\mathcal{E}%
\left( Z^{\xi }\right) _{T}}{D_{T}^{\mathbb{Q}}}\right) \right] \right\}
+\vartheta \left( \mathbb{Q}\right) \right\} .%
\end{array}
\label{Model_1_robust dual value function (Eqt)}
\end{equation}
\end{theorem}

\begin{proof}
Condition $\left( \ref{vQ(y)<oo_for_y>0_Q_nice}\right) $, together with
Lemma 4.4 in \cite{Schd 2007} and Theorem 2 in \cite{Krk&Scha 2003} , imply
the following identity%
\begin{equation*}
v\left( y\right) =\inf\limits_{\mathbb{Q}\in \mathcal{Q}_{\approx
}^{\vartheta }}\left\{ \inf\nolimits_{\widetilde{\mathbb{Q}}\in \mathcal{Q}%
_{elmm}\left( \mathbb{Q}\right) }\left\{ \mathbb{E}_{\mathbb{Q}}\left[
V\left( yd\widetilde{\mathbb{Q}}/d\mathbb{Q}\right) \right] \right\}
+\vartheta \left( \mathbb{Q}\right) \right\} .
\end{equation*}%
Since $\mathcal{D}_{elmm}\left( \mathbb{P}\right) \subset \mathcal{C}$, we
get 
\begin{equation}
\begin{array}{cl}
v\left( y\right) & \geq \inf\limits_{\mathbb{Q}\in \mathcal{Q}_{\approx
}^{\vartheta }}\left\{ \inf\limits_{\xi \in \mathcal{C}}\left\{ \mathbb{E}_{%
\mathbb{Q}}\left[ V\left( y\frac{\mathcal{E}\left( Z^{\xi }\right) _{T}}{%
D_{T}^{\mathbb{Q}}}\right) \right] \right\} +\vartheta \left( \mathbb{Q}%
\right) \right\} \\ 
& \geq \inf\limits_{\mathbb{Q}\in \mathcal{Q}_{\ll }}\left\{
\inf\nolimits_{\xi \in \mathcal{C}}\left\{ \mathbb{E}_{\mathbb{Q}}\left[
V\left( y\tfrac{\mathcal{E}\left( Z^{\xi }\right) _{T}}{D_{T}^{\mathbb{Q}}}%
\right) \right] \right\} +\vartheta \left( \mathbb{Q}\right) \right\} .%
\end{array}
\label{robust dual value function >(Ineqt)}
\end{equation}%
Finally, from Lemma 4.2 in Schied \cite{Schd 2007} and $\mathcal{C}\subset 
\mathcal{Y}_{\mathbb{P}}\left( 1\right) $ follows 
\begin{equation*}
v_{\mathbb{Q}}\left( y\right) \leq \inf_{\xi \in \mathcal{C}}\left\{ \mathbb{%
E}_{\mathbb{Q}}\left[ V\left( y\frac{\mathcal{E}\left( Z^{\xi }\right) _{T}}{%
D_{T}^{\mathbb{Q}}}\right) \right] \right\} ,
\end{equation*}%
and we have the inequalities $\left( \ref{robust dual value function
>(Ineqt)}\right) $ in the other direction, and the result follows.
\end{proof}

\subsection{The logarithmic utility case\label{Sect:_log_utility_case}}

As it was pointed out above in Example \ref{Log-utility:_dominated_by_power}%
, the existence of a solution to the dual problem for the logarithmic
utility function $U\left( x\right) =\log \left( x\right) $ can be read from
the results presented in the previous subsection. However, the nature of the
optimization problem arising in the case of a logarithmic utility deserves a
deeper study. Let $h,h_{0}$ and $h_{1}$ be as in Subsection \ref%
{MinimalPenalties}, satisfying also the following growth conditions: 
\begin{eqnarray*}
h\left( x\right) &\geq &x, \\
h_{0}\left( x\right) &\geq &\frac{1}{2}x^{2}, \\
h_{1}\left( x\right) &\geq &\left\{ \left\vert x\right\vert \vee x\ln \left(
1+x\right) \right\} \mathbf{1}_{\left( -1,0\right) }\left( x\right) +x\left(
1+x\right) \mathbf{1}_{\mathbb{R}_{+}}\left( x\right) .
\end{eqnarray*}%
Now, define the penalization function 
\begin{equation}
\begin{array}{rl}
\vartheta _{\log }\left( \mathbb{Q}\right) := & \mathbb{E}_{\mathbb{Q}}\left[
\int\limits_{0}^{T}h\left( h_{0}\left( \theta _{0}\left( t\right) \right)
+\int\nolimits_{\mathbb{R}_{0}}h_{1}\left( \theta _{1}\left( t,x\right)
\right) \nu \left( dx\right) \right) dt\right] \mathbf{1}_{\mathcal{Q}_{\ll
}}\left( \mathbb{Q}\right) \\ 
& +\infty \times \mathbf{1}_{\mathcal{Q}_{cont}\setminus \mathcal{Q}_{\ll
}}\left( \mathbb{Q}\right) .%
\end{array}
\label{Theta_log_utility}
\end{equation}

\begin{remark}
\label{penalty(Q)<oo_=>_integrability} Notice that when $\mathbb{Q}\in 
\mathcal{Q}_{\ll }^{\vartheta _{\log }}(\mathbb{P})$ with coefficient $%
\theta =\left( \theta _{0},\theta _{1}\right) $ has a finite penalization,
the following $\mathbb{Q}$-integrability properties hold: 
\begin{equation*}
\begin{array}{ll}
\left( \ref{penalty(Q)<oo_=>_integrability}.i\right) & \int\limits_{\left[
0,T\right] \times \mathbb{R}_{0}}\theta _{1}\left( t,x\right) \mu _{\mathbb{P%
}}^{\mathcal{P}}\left( dt,dx\right) \in \mathcal{L}^{1}\left( \mathbb{Q}%
\right) \\ 
\left( \ref{penalty(Q)<oo_=>_integrability}.ii\right) & \int\limits_{\left[
0,T\right] \times \mathbb{R}_{0}}\left\{ 1+\theta _{1}\left( t,x\right)
\right\} \ln \left( 1+\theta _{1}\left( t,x\right) \right) \mu _{\mathbb{P}%
}^{\mathcal{P}}\left( dt,dx\right) \in \mathcal{L}^{1}\left( \mathbb{Q}%
\right) \\ 
\left( \ref{penalty(Q)<oo_=>_integrability}.iii\right) & \int\limits_{\left[
0,T\right] \times \mathbb{R}_{0}}\ln \left( 1+\theta _{1}\left( s,x\right)
\right) \mu \left( ds,dx\right) \in \mathcal{L}^{1}\left( \mathbb{Q}\right)
\\ 
\left( \ref{penalty(Q)<oo_=>_integrability}.iv\right) & \mathbb{E}_{\mathbb{Q%
}}\left[ \tint\limits_{]0,T]\times \mathbb{R}_{0}}\ln \left( 1+\theta
_{1}\right) d\mu \right] =\mathbb{E}_{\mathbb{Q}}\left[ \tint\limits_{]0,T]%
\times \mathbb{R}_{0}}\left\{ \ln \left( 1+\theta _{1}\right) \right\}
\left( 1+\theta _{1}\right) d\mu _{\mathbb{P}}^{\mathcal{P}}\right] \\ 
\smallskip &  \\ 
\smallskip & \text{In addition, for }\mathbb{Q}\in \mathcal{Q}_{\approx
}^{\vartheta _{\log }}(\mathbb{P})\text{ we have} \\ 
\smallskip &  \\ 
\left( \ref{penalty(Q)<oo_=>_integrability}.v\right) & \int\limits_{\left[
0,T\right] \times \mathbb{R}_{0}}\theta _{1}\left( s,x\right) \mu \left(
ds,dx\right) <\infty \ \mathbb{P}-a.s.%
\end{array}%
\end{equation*}
\end{remark}

For $\mathbb{Q}\in \mathcal{Q}_{\ll }(\mathbb{P})$, the relative entropy
function is defined as 
\begin{equation*}
H(\mathbb{Q}|\mathbb{P})=\mathbb{E}\left[ D_{T}^{\mathbb{Q}}\log \left(
D_{T}^{\mathbb{Q}}\right) \right] .
\end{equation*}

\begin{lemma}
\label{H(P|Q)<Penalty(Q)} Given $\mathbb{Q\in }\mathcal{Q}_{\approx
}^{\vartheta _{\log }}\left( \mathbb{P}\right) $, it follows that 
\begin{equation*}
H\left( \mathbb{Q}\left\vert \mathbb{P}\right. \right) \leq \vartheta _{\log
}\left( \mathbb{Q}\right) .
\end{equation*}
\end{lemma}

\begin{proof}
For $\mathbb{Q\in }\mathcal{Q}_{\approx }^{\vartheta _{\log }}\left( \mathbb{%
P}\right) $, Remark \ref{penalty(Q)<oo_=>_integrability} implies that 
\begin{eqnarray*}
H\left( \mathbb{Q}\left\vert \mathbb{P}\right. \right) &=&\mathbb{E}_{%
\mathbb{Q}}\left[ \frac{1}{2}\int\nolimits_{0}^{T}\left( \theta _{0}\right)
^{2}ds+\int\limits_{]0,T]\times \mathbb{R}_{0}}\ln \left( 1+\theta
_{1}\left( s,x\right) \right) \mu \left( ds,dx\right)
-\int\limits_{0}^{T}\int\limits_{\mathbb{R}_{0}}\theta _{1}\left( s,x\right)
\nu \left( dx\right) ds\right] \\
&\leq &\mathbb{E}_{\mathbb{Q}}\left[ \int\limits_{0}^{T}\left\{ \frac{1}{2}%
\left( \theta _{0}\right) ^{2}ds+\int\limits_{\mathbb{R}_{0}}\left\{ \ln
\left( 1+\theta _{1}\left( s,x\right) \right) \right\} \theta _{1}\left(
s,x\right) \nu \left( dx\right) \right\} ds\right] \\
&\leq &\vartheta _{\log }\left( \mathbb{Q}\right) .
\end{eqnarray*}
\end{proof}

\smallskip\ 

\begin{lemma}
\label{U(x)=log(x)_=>_vQ(y)<oo (Lemma)} Let $U\left( x\right) =\log \left(
x\right) $ and $\vartheta _{\log }$ be as in $\left( \ref{Theta_log_utility}%
\right) $. Then the robust utility maximization problem $\left( \ref%
{robust_probl._primal_value_funct}\right) $ has an optimal solution.
\end{lemma}

\begin{proof}
Again, we only need to verify that condition (\ref{vQ(y)<oo_for_y>0_Q_nice})
holds. Observe that for $\mathbb{Q\in }\mathcal{Q}_{\ll }$ fix we have that 
\begin{equation*}
v_{\mathbb{Q}}\left( y\right) \leq \inf_{\xi \in \mathcal{C}}\left\{ \mathbb{%
E}\left[ D_{T}^{\mathbb{Q}}\log \left( \frac{D_{T}^{\mathbb{Q}}}{\mathcal{E}%
\left( Z^{\xi }\right) _{T}}\right) -\log \left( y\right) -1\right] \right\}
.
\end{equation*}%
Also, Proposition \ref{Q ELMM iff} and the Novikov condition yield for $%
\widetilde{\xi }\in \mathcal{C}$, with $\widetilde{\xi }^{\left( 0\right)
}:=-\dfrac{\alpha _{s}}{\beta _{s}}$ and $\widetilde{\xi }^{\left( 1\right)
}:=0$, that $\widetilde{\mathbb{Q}}\in \mathcal{Q}_{elmm}$, where $d%
\widetilde{\mathbb{Q}}\backslash d\mathbb{P}=D_{T}^{\widetilde{\xi }}:=%
\mathcal{E}\left( Z^{\widetilde{\xi }}\right) _{T}.$ Further, from Lemma \ref%
{H(P|Q)<Penalty(Q)} we conclude for $\mathbb{Q\in }\mathcal{Q}_{\approx
}^{\vartheta _{\log }}\left( \mathbb{P}\right) $ that 
\begin{equation*}
\mathbb{E}\left[ D_{T}^{\mathbb{Q}}\log \left( \frac{D_{T}^{\mathbb{Q}}}{%
D_{T}^{\widetilde{\xi }}}\right) \right] =H\left( \mathbb{Q}\left\vert 
\mathbb{P}\right. \right) +\mathbb{E}_{\mathbb{Q}}\left[ \int\limits_{0}^{T}%
\frac{\alpha _{s}}{\beta _{s}}\theta _{s}^{\left( 0\right) }ds+\frac{1}{2}%
\int\limits_{0}^{T}\left( \dfrac{\alpha _{s}}{\beta _{s}}\right) ^{2}ds%
\right] <\infty
\end{equation*}%
and the claim follows.
\end{proof}

\end{document}